\documentclass[conference]{IEEEtran}
\IEEEoverridecommandlockouts
\usepackage{amsthm} 
\newtheorem{theorem}{Theorem}

\newtheorem{lemma}{Lemma}

\newtheorem{assumption}{Assumption}
\usepackage{algorithm}
\usepackage{algpseudocode} 
\usepackage{amsmath,amsfonts}
\allowdisplaybreaks[4]
\usepackage{array}
\usepackage[caption=false,font=normalsize,labelfont=sf,textfont=sf]{subfig}
\usepackage{textcomp}
\usepackage{stfloats}
\usepackage{url}
\usepackage{verbatim}
\usepackage{graphicx}
\usepackage{cite}
\usepackage{amssymb}
\usepackage{nomencl}
\makenomenclature
\usepackage{etoolbox}
\usepackage{url}
\usepackage{booktabs} 
\usepackage{xcolor}

\def\BibTeX{{\rm B\kern-.05em{\sc i\kern-.025em b}\kern-.08em
    T\kern-.1667em\lower.7ex\hbox{E}\kern-.125emX}}
\begin{document}

\title{A mixed $\mathcal{H}_{\infty}$--Passivity approach for Leveraging District Heating Systems as Frequency Ancillary Service in Electric Power Systems
}

\author{Xinyi Yi and Ioannis Lestas%
    \thanks{X. Yi and I. Lestas are with the Department of Engineering, University of Cambridge, Trumpington Street, Cambridge, CB2 1PZ, United Kingdom. Emails:
        {\tt\small <xy343, icl20>@cam.ac.uk}.}
    \thanks{Accepted by 2026 IEEE Power \& Energy Society General Meeting}%
}

\maketitle

\begin{abstract}
This paper introduces a mixed $\mathcal{H}_{\infty}$–passivity framework that enables district heating systems (DHSs) with heat pumps to support electric-grid frequency regulation. The analysis in the paper illustrates how the DHS regulator influences coupled electro–thermal frequency dynamics and provides LMI conditions for efficient controller design. A disturbance-independent temperature regulator is also presented that ensures stability and robustness against heat-demand uncertainty. Simulations demonstrate improved  frequency control dynamics in the electrical power grid, while maintaining good thermal performance in the DHS.
\end{abstract}

\section{Introduction}
As decarbonization accelerates, electrified heating is becoming central to low-carbon energy systems. Heat pumps (HPs) enable efficient electricity-to-heat conversion and significantly tighten the coupling between electric power systems (EPSs) and district heating systems (DHSs). Large-scale initiatives including the EU Heat Pump Accelerator~\cite{heatpumpaccel} and China’s National Action Plan~\cite{chinaheatpump} highlight the rapid growth of HP deployment. At the same time, the expansion of large data centers increases the need for fast and reliable ancillary services, intensifying the demand for flexible resources. Leveraging their inherent thermal storage capacity, DHSs present strong potential to contribute to EPS frequency support.

Frequency regulation in EPSs is hierarchical: primary control acts within seconds, secondary restores nominal frequency within minutes, and tertiary optimizes long-term operation. The thermal inertia of DHSs fits the secondary-control time scale. DHSs typically use optimal setpoint scheduling and tracking, where energy-sharing optimization yields cost-efficient setpoints and controllers follow them. Yet uncertain thermal demand causes setpoint deviations. Our prior work~\cite{accconfer} proposed an  LQR temperature regulator achieving optimal power sharing within a heating system without having a prior knowledge of the disturbances. The integration of such involved control schemes within a power system,  while ensuring good performance in the combined  electrothermal network is though a non-trivial problem.

Recent studies have examined coordinated frequency– temperature control in combined heat and power (CHP) systems. \cite{krishna2021port} introduced a Port-Hamiltonian framework, but it depends on disturbance-based temperature setpoints that are hard to specify under renewable variability and complex demand patterns. \cite{qin2024frequency} proposed a distributed primary-frequency controller with a disturbance-independent average-temperature based regulator, while \cite{ref:article5} used robust MPC with data-driven disturbance forecasts to provide frequency support under building-scale thermal constraints. However, prior work does not assess how the the DHS temperature-regulator design affects closed-loop electro–thermal dynamic interactions and performance  
when HPs contribute to frequency regulation. This study addresses this gap using a mixed $\mathcal{H}_{\infty}$–passivity framework that links DHS regulator structure to CHP frequency control performance.

In particular, a mixed $\mathcal{H}_{\infty}$--passivity method is proposed that ensures stability and enhances performance, while achieving optimal energy sharing in the DHS. Furthermore, it yields tractable LMI conditions for the DHS temperature controller design. The key contributions of this paper are:  
1) A DHS regulator is proposed that converges to the optimal electro--thermal operating point under uncertain demand;  
2) A mixed $\mathcal{H}_{\infty}$--passivity framework is formulated linking DHS regulator structure to CHP frequency-control dynamic behaviour, enabling more efficient frequency support;  
3) Analytical convergence guarantees are provided for the CHP system considered. The paper is structured as follows. Section~II formulates the DHS--EPS model, Section~III develops and analyzes the control scheme,  
and Section~IV presents simulation case studies.

\section{System Model Formulation}
In the paper, bold symbols denote stacked vectors, and when clear from context, matrices. Non-bold symbols represent scalar quantities. All proofs are provided in the Appendix to preserve technical clarity and comply with page limits.

\subsection{Heat pump}
Let $H_e$ denote the set of buses in the EPS where HPs are connected, $H_h$ the set of DHS edges associated with HPs, and $H_{HP}$ the set of HP units.  
For frequency regulation, each HP operates near an equilibrium point with approximately constant coefficient of performance (CoP), so that its electric power consumption $p^{H}_{k}$, heat output $h^{P}_{k}$, bus frequency deviation $\omega_i$ and DHS signal $Sig^H$ satisfy\cite{qin2024frequency,krishna2021port,ref:article3}
:
\begin{subequations}\label{firstorder}
\begin{align}
    h^{P}_{j_k} = &CoP_k \cdot p^{H}_{i_k}, j_k \in H_h, i_k \in H_e,\label{heatpumpcop}\\
p_k^H=&\gamma_{k}^E \omega_k+ \gamma_{k}^H Sig^H , k \in H_{HP},\label{heatpumpcopcal}
\end{align}
\end{subequations}
where $\gamma_{k}^E$ and $\gamma_{k}^H$ denote the electrical and thermal operation coefficients of HP $k$, respectively.
\subsection{Electric power system}
\subsubsection{Frequency dynamic model of electric buses}
Two types of buses are considered: $i \in \mathcal{E}^{DG}$, which connect only to distributed generators (DGs), and $i \in \mathcal{E}^{HP}$, which host both DGs and heat pumps (HPs), and $\mathcal{E} = \mathcal{E}^{DG} \cup \mathcal{E}^{HP}$.
Let $P_i$ denote the net active power flowing out of bus $i$. 
The frequency deviation, damping coefficient, controllable generation, control input, electrical disturbance, and inertia at bus $i$ are denoted by $\omega_i$, $D_i$, $P_i^G$, $u_i$, $P_i^L$, and $M_i$, respectively. 
Here, $\omega_i$, $u_i$, $P_i^L$, $P_i$, and $P_i^G$ represent deviations from their nominal steady-state values.
To capture the lag in the response of generation, each controllable unit is modeled with a first-order system:
\begin{equation}\label{gen_dyn}
    T_{g,i}\,\dot P_i^G = -P_i^G + u_i, \qquad i \in \mathcal{E},
\end{equation}
where $T_{g,i}>0$ is the effective turbine--governor time constant.

The bus frequency dynamics are then given by
\begin{subequations}\label{busfre1}
\begin{align}
    M_i \dot{\omega}_i &= -D_i \omega_i - P_i + P_i^G, 
    && i \in \mathcal{E}^{DG}, \label{busfrea}\\[2mm]
    M_i \dot{\omega}_i &= -D_i \omega_i - P_i + P_i^G - p_i^H, 
    && i \in \mathcal{E}^{HP}, \label{hpdhs}
\end{align}
\end{subequations}
where $p_i^H$ is the active power consumed by the HP at bus $i$, computed according to~(\ref{heatpumpcopcal}).
The network power flow and angle dynamics are:
\begin{equation}\label{busfre2}
P_i = \sum_j |B_{ij}| V_i V_j \sin(\theta_i - \theta_j) - P_i^L,
     i \in \mathcal{E}, 
    \dot{\boldsymbol{x_\theta}} = \boldsymbol{R_I \omega}, 
\end{equation}
where $\boldsymbol{x}_\theta$ is the phase-angle deviation vector excluding the reference bus $r$, $\boldsymbol{\omega}$ denotes frequency vector, and $\boldsymbol{R_I = [I_{\,n^{\mathcal{E}}-1},\, -\mathbf{1}_{\,n^{\mathcal{E}}-1}]}$.  $n^{\mathcal{E}}$ denotes the number of electric buses, while $\boldsymbol{I_x}$ and $\boldsymbol{1_x}$ represent the identity matrix and the all-ones vector of dimension $x$, respectively.

\begin{assumption}
For any $\boldsymbol{\bar{P}^L}$ and  $\boldsymbol{\bar{p}^H}$, the EPS~(\ref{gen_dyn}-\ref{busfre2}) admits an equilibrium $(\boldsymbol{x_\theta^*}, \boldsymbol{\theta^*}, \boldsymbol{\omega^*}, \boldsymbol{P^{G*}}, \boldsymbol{P^*},\boldsymbol{u^*})$,  
where $\boldsymbol{\omega^*} = \omega^{com}\boldsymbol{1}_{n^{\mathcal{E}}}$ and $\omega^{com}$ is the synchronized EPS frequency.
\end{assumption}

\subsubsection{Secondary frequency control}
The objective of secondary frequency control is to restore the frequency to
its nominal value, achieved via AGC integral control at bus $r$:
\begin{equation}\label{eq:agc}
    \dot{g} = -\omega_r, \quad
    u_i = -K_i^{P}\,\omega_i + \delta_{ir} K^{I} g ,\quad i\in\mathcal{E},
\end{equation}
where 
$K_i^{P}>0$ are droop
gains, and $K^{I}>0$ is the AGC gain. If $i=r$, $\delta_{ir}=1$, otherwise, $\delta_{ir}=0$.

\begin{lemma}
Under \textbf{Assumption~1}, for any constant disturbances 
$\boldsymbol{\bar{P}^{L}}$ and $\boldsymbol{\bar{p}^{H}}$,  
the EPS model with the AGC controller \eqref{gen_dyn}-\eqref{eq:agc} admits a unique equilibrium  
$(\boldsymbol{x_\theta^*}, \boldsymbol{\theta^*}, \boldsymbol{u^*}, \boldsymbol{\omega^*}, \boldsymbol{P^{G*}}, \boldsymbol{P^*},g^*)$.  
At this equilibrium, $\boldsymbol{\omega^{*}=0}$.
\end{lemma}

\begin{assumption}
For all $i\in\mathcal{E}$ the damping coefficients satisfy $D_i \ge D_{\min}>0$.  
The closed–loop generation dynamics (\ref{gen_dyn},\ref{eq:agc}) at each bus are
of first order with time constant $T_{g,i}>0$.  
The transfer function from $-\omega_i$ to the generated power $P_i^G$ is
\begin{equation}\label{eq:G_i_def}
   G_i(s)
   := 
   D_i
   \;+\;
   \frac{K_i^{P}}{T_{g,i}s+1}
   \;+\;
   \delta_{ir}\,\frac{K^{I}}{s\,(T_{g,i}s+1)} ,
\end{equation}
where the last term is present only at the reference bus $r$.
Each $G_i(s)$ is strictly positive real, i.e., there exists $\rho_i>0$
such that $\Re\,G_i(j\omega)\ge \rho_i$ for all $\omega\in\mathbb{R}$. \footnote{Note that higher order generation dynamics that satisfy this property can also be considered.}
\end{assumption}

\begin{lemma}
Under \textbf{Assumptions 1-2}, consider the EPS model with the AGC controller \eqref{gen_dyn}-\eqref{eq:agc}.  
With the input $\boldsymbol{u_1 := -p^H}$,
and the output $\boldsymbol{y_1 := \omega^{HP}}$, the EPS is strictly passive from $\boldsymbol{u_1}$ to
$\boldsymbol{y_1}$: there exist a storage function $V_{\mathrm{e}}$ 
and a constant $\rho_e>0$ such that along all trajectories with the squared Euclidean norm $\|\boldsymbol{y_1}\|^2 := \boldsymbol{y_1^\top y_1}$:
\begin{equation}\label{eq:eps_passive_ineq}
\dot V_{\mathrm{e}}
    \;\le\; 
    \boldsymbol{y_1^\top u_1} - \rho_e\,\|\boldsymbol{y_1}\|^2.
\end{equation}
\end{lemma}

\subsection{District heating system}
\subsubsection{Temperature dynamics}
We model the DHS temperature dynamics following~\cite{qin2024frequency}.  
Let $\mathcal{H}^E$ denote the set of edges (heat exchangers and pipelines) and $\mathcal{H}^N$ the set of nodes (storage tanks).  
The edge and node dynamics are given as follows:
\begin{subequations}\label{temdy}
\begin{align}
\rho C_p V_j^E \dot{T}_j^E =& \rho C_p q_j^E (T_k^N - T_j^E) + h_j^G + h_j^P - h_j^L,\notag\\
&j \in \mathcal{H}^E,\; k \in \mathcal{H}^N, \label{eq:2a} \\
\rho C_p V_k^N \dot{T}_k^N =& \sum_{j \in \mathcal{T}_k} \rho C_p q_j^E (T_j^E - T_k^N), 
k \in \mathcal{H}^N, \label{eq:2b}
\end{align}
\end{subequations}
where $T_j^E$ and $T_k^N$ denote the outlet temperatures of edge $j$ and node $k$, respectively.  
$h_j^P$, $h_j^G$, and $h_j^L$ represent the heat contributions from the HP, the conventional heat source, and the thermal load at edge $j$.  
The parameters $V_j^E$ and $V_k^N$ are the volumes of edges and nodes, $q_j^E$ is the mass flow along edge $j$, and $C_p$ and $\rho$ denote the specific heat capacity and density of water.  
The HP heat injection $h_j^P$ corresponds to $h_{j_k}^P$ in~\eqref{heatpumpcop}.  
If edge $j$ hosts an HP, then $h_j^G=h_j^L=0$;  
for a conventional heat source, $h_j^P=h_j^L=0$;  
and for a load, $h_j^G=h_j^P=0$ with fixed $h_j^L$.
Equations~\eqref{eq:2a}–\eqref{eq:2b} can be compactly expressed in matrix form using the  Kirchhoff matrix $\boldsymbol{A_h}$ \cite{qin2024frequency}.
\small{\begin{subequations}
\begin{align}
&\boldsymbol{V}
\begin{bmatrix}
\dot{\boldsymbol{T}}^G \\
\dot{\boldsymbol{T}}^{HP}\\
\dot{\boldsymbol{T}}^{L}\\
\dot{\boldsymbol{T}}^N
\end{bmatrix}
= -\boldsymbol{A}_h 
\begin{bmatrix}
\boldsymbol{T}^G \\
\boldsymbol{T}^{HP}\\
\boldsymbol{T}^{L}\\
\boldsymbol{T}^N
\end{bmatrix}
+ 
\begin{bmatrix}
\boldsymbol{h^G} \\
\boldsymbol{0}\\
\boldsymbol{0} \\
\boldsymbol{0}
\end{bmatrix}
+ 
\begin{bmatrix}
\boldsymbol{0} \\
\boldsymbol{h^{P}}\\
\boldsymbol{0}\\
\boldsymbol{0}
\end{bmatrix}
+ 
\begin{bmatrix}
\boldsymbol{0}\\
\boldsymbol{0}\\
-\boldsymbol{h^{L}}\\
\boldsymbol{0}
\end{bmatrix}\label{eq:3},\\
&\boldsymbol{\dot{T}}=\boldsymbol{-AT+B_1h^G+B_2h^{P}+w^h},\label{origindy}
\end{align}
\end{subequations}}

\noindent where $\boldsymbol{h^G}$, $\boldsymbol{h^{P}}$, and $\boldsymbol{h^L}$ denote the non-pump heat source generation, HP generation, and load vectors, respectively.  
We define
$\boldsymbol{A = \tfrac{A_h}{V}}, 
\boldsymbol{B_1 = \tfrac{I}{V}\begin{bmatrix} I \\ 0 \\ 0 \\ 0 \end{bmatrix}}, 
\boldsymbol{B_2 = \tfrac{I}{V}\begin{bmatrix} 0 \\ I \\ 0 \\ 0 \end{bmatrix}}, 
\boldsymbol{w^h = \tfrac{I}{V}\begin{bmatrix} 0 \\ 0 \\ -h^L \\ 0 \end{bmatrix}}$.

\subsubsection{Energy-sharing steady-state objectives}
The steady-state economic dispatch for DHSs under given $\boldsymbol{\bar{h}^{P}}$ and $\boldsymbol{\bar{w}^h}$ is modeled as two optimization problems.
\small{\begin{subequations}\label{opt1}
\begin{align}
\textbf{E1:} &\ \min_{\boldsymbol{h^G}\in \mathbb{R}, \boldsymbol{T}\in \mathbb{R}} \Phi_1 = \frac{1}{2} \boldsymbol{{h^G}^{\top}} \boldsymbol{F^G} \boldsymbol{h^G},\\
    &\ \text{s.t.} \ \boldsymbol{A T} =\boldsymbol{B_1 h^G}+\boldsymbol{B_2 \bar{h}^{P}+\bar{w}^h},\label{5b}
\end{align}
\end{subequations}}

\noindent where $\Phi_1$ denotes the non-pump sources' heat generation cost, with 
\(\boldsymbol{F^G} = \mathrm{diag}\{ f_i^G \}\), where \(f_i^G > 0\) is the cost coefficient of source~\(i\).  
The temperature deviation cost is minimized in \textbf{E2}, formulated as:
\small{\begin{subequations}\label{economic2tran}
\begin{align}
&\textbf{E2:} \min_{z \in \mathbb{R}, \boldsymbol{T \in \mathbb{R}}} \Phi_2=\frac{1}{2} \boldsymbol{T}^\top \boldsymbol{F^D} \boldsymbol{T},\\
&s.t. \boldsymbol{T} =\boldsymbol{A^\dagger} ( \boldsymbol{{B_1h^G}^*}+\boldsymbol{B_2\bar{h}^{P}}+\boldsymbol{\bar{w}^h}) + \alpha\boldsymbol{1},\label{11b}
\end{align}
\end{subequations}}

\noindent where $\boldsymbol{A^\dagger}$ denotes the Moore--Penrose pseudoinverse of $\boldsymbol{A}$, and 
$\boldsymbol{F^D} = \mathrm{diag}\{F^D_i\}$, with $F^D_i> 0$, is the temperature deviation penalty coefficient associated with node or edge~$i$.

\begin{lemma}
\textbf{(Optimality condition for \textbf{E1} and \textbf{E2}\cite{accconfer})}
If the DHS (\ref{origindy}) achieves equilibrium at $\boldsymbol{{T^*}}$ and $\boldsymbol{{h^G}^*}$, and satisfies $\boldsymbol{F^M h^{G*}} = \boldsymbol{0}$ and $\boldsymbol{1^\top F^D {T^*}} = 0$, then it uniquely\footnote{Because $\Phi_1$ and $\Phi_2$ are strictly convex
($\boldsymbol F^D\succ 0$, $\boldsymbol F^G\succ 0$) under linear constraints,
the optimizer $(\boldsymbol T^*,\boldsymbol h^{G*})$ of \textbf{E1}--\textbf{E2}
is unique.} solves the optimization problems \textbf{E1} and \textbf{E2}, where $\boldsymbol{F^M}$
is defined by the following matrix: 
\small{\begin{equation}
\begin{bmatrix}
F^G(1,1) & -F^G(2,2) & 0 & \cdots & 0\\
0 & F^G(2,2) & -F^G(3,3) & \cdots & 0\\
\vdots & \vdots & \vdots & \ddots & \vdots\\
\end{bmatrix}.
\end{equation}}
\end{lemma}
The error is defined as:
\small{\begin{equation}\label{errorde}
\boldsymbol{e} = \boldsymbol{C}\boldsymbol{T} + \boldsymbol{D}\boldsymbol{h^G}=\begin{bmatrix}
  \boldsymbol{e^G}\\Sig^H
\end{bmatrix},
Sig^H=\boldsymbol{s_Te}, .
\end{equation}}

\noindent where $\boldsymbol{C} = \begin{bmatrix} \boldsymbol{0} \\ \boldsymbol{1^\top F^D} \end{bmatrix}$, $\boldsymbol{D} = \begin{bmatrix} \boldsymbol{F^M} \\ 0 \end{bmatrix}$, and $\boldsymbol{s_T}=\begin{bmatrix}\boldsymbol 0^\top& 1\end{bmatrix}$

\subsubsection{Augmented DHS} Substituting (\ref{firstorder}) into (\ref{origindy}):
\begin{subequations}\label{originaldhs}
\begin{align}
    \boldsymbol{\dot{T}}
    =&\boldsymbol{(-A+B_sC)T+(B_1+B_sD)h^G}\notag\\
    &\boldsymbol{+B_2CoP\gamma^E w^{HP}} \boldsymbol{+w^h}\label{originaldhsd}\\
    =&\boldsymbol{A_e T+B_e h^G+B_\omega \omega^{HP}+w^h,}
\end{align}
\end{subequations}

\noindent where $\boldsymbol{CoP=}diag(CoP_i)$, $\boldsymbol{B_\omega=B_2CoP\gamma^E}$, and $\boldsymbol{B_s=B_2CoP\gamma^H 1_{n^{HP}}s_T}$. Integrating $\dot{\boldsymbol \xi}=\boldsymbol e$ obtains:
\small{\begin{equation}\label{augmentdhs}
    \begin{bmatrix}
        \boldsymbol{\dot{T}}\\
        \boldsymbol{\dot{\xi}}
    \end{bmatrix}=\underbrace{\begin{bmatrix}
        \boldsymbol{A_e}&\boldsymbol{0}\\
        \boldsymbol{C}&\boldsymbol{0}
    \end{bmatrix}}_{\triangleq \boldsymbol{A_\textit{Aug}}} \begin{bmatrix}
        \boldsymbol{T}\\
        \boldsymbol{\xi}
    \end{bmatrix}+\underbrace{\begin{bmatrix}
        \boldsymbol{B_e}\\\boldsymbol{D}
    \end{bmatrix}}_{\triangleq \boldsymbol{B_\textit{Aug}}}\boldsymbol{h^G+}\underbrace{\begin{bmatrix}\boldsymbol B_\omega\\ \boldsymbol 0\end{bmatrix}}_{\triangleq \boldsymbol{B_\textit{cl}^{(\omega)}}}\boldsymbol{\omega^{HP}}+\underbrace{\begin{bmatrix}\boldsymbol I\\ \boldsymbol 0\end{bmatrix}}_{\triangleq \boldsymbol{B_\textit{cl}^{(h)}}}\boldsymbol{w^h}.
\end{equation}}With $\boldsymbol x=\begin{bmatrix}\boldsymbol T^\top & \boldsymbol \xi^\top\end{bmatrix}^\top$, (\ref{augmentdhs}) can be rewritten as:
\begin{equation}\label{augdhs}
    \boldsymbol{\dot{x}=A_\textit{aug}x+B_\textit{aug}h^G+B_\textit{cl}^{(\omega)} \omega^{HP}+B_\textit{cl}^{(h)}w^h}.
\end{equation}
The input is $\boldsymbol{u_2=\omega^{HP}}$ and the output is defined as $\boldsymbol{y_2=p^H}$:
\small{\begin{equation}\label{errorde2}
\begin{aligned}
    \boldsymbol{p^H}=&\boldsymbol{\gamma^E\omega^{HP}}+\boldsymbol{\gamma^H}Sig^H\boldsymbol{1_{n^{HP}}}\\
    =&\boldsymbol{\gamma^E \omega^{HP}}+\underbrace{\boldsymbol{\gamma^H}\begin{bmatrix}
     \boldsymbol{1_{n^{HP}}s_TC}&\boldsymbol{0}   
    \end{bmatrix}}_{\triangleq \boldsymbol{S_C}}\boldsymbol{x}\\
    &-\underbrace{\boldsymbol{\gamma^H 1_{n^{HP}}s_T D}}_{\triangleq\boldsymbol{S_D}} \boldsymbol{Kx}\\
    =&\boldsymbol{\gamma^E \omega^{HP}+}C_y(\boldsymbol{K}) \boldsymbol{x}.
\end{aligned}
\end{equation}}

\begin{lemma}
Consider the augmented DHS dynamics in~\eqref{augmentdhs} with the temperature regulator $\boldsymbol{h^G} = -\boldsymbol{K_T T} - \boldsymbol{K_I \xi}$. The resulting closed-loop temperature dynamics are
\begin{equation}\label{closeloopdhs}
\dot{\boldsymbol{x}}
=
\boldsymbol{A}_{\rm cl}\,\boldsymbol{x}
+\boldsymbol{B}_{\rm cl}^{(\omega)}\,\boldsymbol{\omega^{HP}}
+\boldsymbol{B}_{\rm cl}^{(h)}\,\boldsymbol{w^{h}},
\end{equation}
where $\boldsymbol{A}_{\rm cl}=
\begin{bmatrix}
\boldsymbol{A_e}-\boldsymbol{B_e K_T}
&
-\boldsymbol{B_e K_I}
\\[2pt]
\boldsymbol{C}-\boldsymbol{D K_T}
&
-\boldsymbol{D K_I}
\end{bmatrix}$.
If $\boldsymbol{A}_{\rm cl}$ is Hurwitz, the closed-loop DHS admits a unique equilibrium. Moreover, for constant disturbances $\boldsymbol{\bar \omega^{HP}}$ and $\boldsymbol{\bar w^h}$, 
$\lim_{t\to\infty} \boldsymbol{e}(t) = \boldsymbol{0}$.
\end{lemma}

\section{PASSIVITY-BASED FREQUENCY AND
TEMPERATURE CONTROL}
The EPS, with input $\boldsymbol{u_1}$ and output $\boldsymbol{y_1}$, and the DHS, with input $\boldsymbol{u_2}$ and output $\boldsymbol{y_2}$, are interconnected in the CHP system through the relations  
$\boldsymbol{u_1} = -\boldsymbol{y_2}$ and $\boldsymbol{u_2} = \boldsymbol{y_1}$.

\subsection{Stability and optimality of the CHP system}
\begin{lemma}
Under \textbf{Assumption~1}, the unique CHP equilibrium  
$(\boldsymbol{x_\theta^*}, \boldsymbol{\theta^*}, \boldsymbol{u^*}, \boldsymbol{\omega^*}, \boldsymbol{P^{G*}}, \boldsymbol{P^*}, g^*, \boldsymbol{h^{G*}}, \boldsymbol{\xi^*}, \boldsymbol{T^*}, \boldsymbol{e^*}, \boldsymbol{h^{P*}}, \boldsymbol{p^{H*}})$  
satisfies $\boldsymbol{\omega^* = 0}$ and $\boldsymbol{e^* = 0}$.  
Furthermore, the equilibrium meets the control objectives of both subsystems:  
$\lim_{t\to\infty} \omega_i = 0$ for all $i \in \mathcal{E}$ (\textbf{Lemma~1});  
and $\lim_{t\to\infty} \boldsymbol{T} = \boldsymbol{T^*}$,  
$\lim_{t\to\infty} \boldsymbol{h^G} = \boldsymbol{h^{G*}}$ (\textbf{Lemma~3}).
\end{lemma}

\begin{theorem}
Under \textbf{Assumptions 1-2}, if the closed-loop augmented DHS in~\eqref{errorde2}--\eqref{closeloopdhs} is strictly passive with respect to the input--output pair $(\boldsymbol{u_2 = \omega^{HP}},\, \boldsymbol{y_2 = p^H})$, and if $\boldsymbol{A_{\rm cl}}$ is Hurwitz, then the CHP equilibrium defined in \textbf{Lemma~5} is asymptotically stable.
\end{theorem}

\subsection{Mixed $H_\infty$/Passivity Controller Design} 
\subsubsection{Passivity-based controller design}
From \textbf{Theorem~1}, the temperature regulator $\boldsymbol{h^G = -Kx}$ must satisfy:  
(i)~$\boldsymbol{A_{\rm cl}}$ is Hurwitz;  
(ii)~the closed-loop augmented DHS~(\ref{errorde2}--\ref{closeloopdhs}) is strictly passive 
with input $\boldsymbol{\omega^{HP}}$ and output $\boldsymbol{p^H} 
= \boldsymbol{C_y(K)x} + \boldsymbol{\gamma^E \omega^{HP}}$.
The system is strictly passive and $\boldsymbol{A_{\rm cl}}$ is Hurwitz 
if there exist $\boldsymbol{P}\succ 0$ and $\rho>0$,
\begin{equation}\label{eq:KYP-BMI-strict}
\begin{bmatrix}
\boldsymbol{A_{\rm cl}}^\top \boldsymbol{P} 
+ \boldsymbol{P}\boldsymbol{A_{\rm cl}}
&
\boldsymbol{P}\boldsymbol{B_{\rm cl}^{(\omega)}} 
- \boldsymbol{C_y(K)}^\top
\\[5pt]
*
&
-(\boldsymbol{\gamma^E} + \boldsymbol{\gamma^E}^\top\,) 
-\rho \boldsymbol I
\end{bmatrix}
\prec \boldsymbol 0.
\end{equation}

To transform original KYP-based passivity BMI~\eqref{eq:KYP-BMI-strict} into LMI for optimization, we introduce the variables
$\boldsymbol{X}:=\boldsymbol{P^{-1}}\succ0$ and $\boldsymbol Y := \boldsymbol K \boldsymbol X$. The KYP-based passivity BMI~\eqref{eq:KYP-BMI-strict} is equivalent to 
\begin{equation}\label{eq:LMI_pass_direct}
\boxed{
\begin{bmatrix}
\boldsymbol{LMI_1} 
& 
\boldsymbol{LMI_2}\\[2pt]
\boldsymbol{LMI_2}^\top 
& 
-\bigl(\boldsymbol\gamma^E + \boldsymbol{{\gamma^E}^\top}\bigr) 
-\rho\,\boldsymbol I
\end{bmatrix}
\ \prec\ \boldsymbol 0,
\boldsymbol X\succ0,\ \rho>0.}
\end{equation}
where $
\boldsymbol{LMI}_1:=\big(\boldsymbol{A}_{\rm aug}\boldsymbol{X}-\boldsymbol{B}_{\rm aug}\boldsymbol{Y}\big)+\big(\boldsymbol{A}_{\rm aug}\boldsymbol{X}-\boldsymbol{B}_{\rm aug}\boldsymbol{Y}\big)^\top$ and
$\boldsymbol{LMI}_2:=\boldsymbol{B}^{(\omega)}_{\rm cl}-\boldsymbol{X}\boldsymbol{S}_C^\top+\boldsymbol{Y}^\top\boldsymbol{S}_D^\top.$

\begin{lemma}
The original DHS~(\ref{originaldhs}) is intrinsically low-pass.  
Let $\boldsymbol{G(s)}$ be the closed-loop transfer matrix from $\boldsymbol{\omega^{HP}}$ to $\boldsymbol{p^H}$ for the augmented DHS~\eqref{errorde2}--\eqref{closeloopdhs}.  
It satisfies  
$\lim_{\omega\to0}\boldsymbol{G(j\omega)}=\boldsymbol{\gamma^E}$.
\end{lemma}

In the low-frequency regime, the DHS behaves as a slowly varying thermal-storage system, so $\boldsymbol{G(j\omega)}$ provides a useful ``virtual damping’’ effect to the EPS. As disturbances reach the DHS time scale (about $50$--$300\,\text{s}$ or $0.003$--$0.02\,\text{Hz}$), pipeline inertia and transport delay introduce significant phase lag, exciting internal diffusive modes rather than attenuating them. Consequently, $\boldsymbol{G(j\omega)}$ loses its damping role and becomes a mid-frequency cross-coupling amplifier, returning stored thermal energy toward the EPS interface and reducing robustness. At high frequencies, the DHS is strongly diffusive, and $\boldsymbol{G(j\omega)}$ acts mainly as a parasitic cross-coupling path; any non-negligible gain in this range increases the transmission of high-frequency disturbances across the EPS–DHS boundary.

This motivates a loop-shaping design that suppresses the \emph{mid- and high-frequency} magnitude of $\boldsymbol{G(j\omega)}$ while preserving its low-frequency behavior. Such shaping prevents excitation of DHS thermal modes and limits unwanted cross-coupling without sacrificing the beneficial low-frequency virtual damping. The goal is therefore to attenuate the mid- and high-frequency components of the channel $\boldsymbol{\omega^{HP}}\!\mapsto\!\boldsymbol{p^H}$ while keeping its low-frequency characteristics essentially unchanged. Reducing $\|\boldsymbol{G(j\omega)}\|$ in this non-damping band mitigates cross-coupling, improves disturbance rejection, and enhances robustness of the DHS–EPS interaction.

\subsubsection{$\mathcal{H}_\infty$ frequency–shaping filter}
\label{subsec:lowpass}
We introduce a high-pass weighting filter for loop–shaping and impose an 
$\mathcal{H}_\infty$ performance constraint on the weighted output.
The high-pass weight is chosen as
$\boldsymbol{W}_{\mathrm{HP}}(s)
=
\alpha+\frac{s/\omega_h}{1+s/\omega_h}, 
0<\alpha\ll1$,
where $\omega_h>0$ is the cutoff frequency of the first-order low-pass factor
$1/(1+s/\omega_h)$.
A state-space realization of the filter is
$\dot{\boldsymbol{x_\omega}}
=
-\omega_h\,\boldsymbol{x_\omega}
+
\omega_h\,\boldsymbol{p^H},\boldsymbol{z}
=
(1+\alpha)\,\boldsymbol{p^H}
-
\boldsymbol{x_\omega}$,
where $\boldsymbol{x_\omega}\in\mathbb{R}_{n^{HP}}$ stores the low-frequency component 
of $\boldsymbol{p^H}$, while $\boldsymbol{z}$ represents the weighted 
high-frequency component to be minimized under the $\mathcal{H}_\infty$ criterion.
The small constant $\alpha$ slightly elevates the low-frequency gain for numerical
robustness and does not affect the cutoff frequency $\omega_h$.
Define the augmented weighted state  
$\boldsymbol{x}_{\mathrm{lp}}
=
\begin{bmatrix}
\boldsymbol{x}^\top & \boldsymbol{x_\omega}^\top
\end{bmatrix}^\top$,  
$(\boldsymbol{\omega^{HP}} \!\to\! \boldsymbol{z})$ 
admits the state-space form:
\small{\begin{subequations}
\begin{align}
\dot{\boldsymbol{x}}_{\mathrm{lp}}
=&
\underbrace{\begin{bmatrix}
\boldsymbol{A}_{\mathrm{cl}}(\boldsymbol{K}) & \boldsymbol{0}\\
\omega_h\,\boldsymbol{C}_y(\boldsymbol{K}) & -\omega_h\,\boldsymbol{I}_{n^{HP}}
\end{bmatrix}}_{\boldsymbol{A}_{\mathrm{lp}}(\boldsymbol{K})}
\boldsymbol{x}_{\mathrm{lp}}
+
\underbrace{\begin{bmatrix}
\boldsymbol{B}_{\mathrm{cl}}^{(\omega)}\\[2pt]
\boldsymbol{0}
\end{bmatrix}}_{\boldsymbol{B}_{\mathrm{lp}}}
\boldsymbol{\omega^{HP}}\notag\\
&+
\begin{bmatrix}
\boldsymbol{B}_{\mathrm{cl}}^{(h)}\\
\boldsymbol{0}
\end{bmatrix}
\boldsymbol{w^{h}},
\\
\boldsymbol{z}
=&
\underbrace{\big[(1+\alpha)\,\boldsymbol{C}_y(\boldsymbol{K}) 
\quad -\,\boldsymbol{I}_{n^{HP}}\big]}_{\boldsymbol{C}_{\mathrm{lp}}(\boldsymbol{K})}
\boldsymbol{x}_{\mathrm{lp}},
\boldsymbol{D}_{\mathrm{lp}}=\boldsymbol{0}.
\end{align}
\end{subequations}}

The Bounded Real Lemma (BRL) provides the sufficient condition for
$\|\mathcal{G}_{\boldsymbol{\omega^{HP}}\to \boldsymbol{z}}\|_\infty < \gamma_\infty$ as the existence of 
$\boldsymbol{P}_{\mathrm{lp}}\succ\boldsymbol{0}$ 
such that
\small{\begin{equation}\label{eq:BRL-lp}
\begin{bmatrix}
\boldsymbol{A}_{\mathrm{lp}}(\boldsymbol{K})^\top \boldsymbol{P}_{\mathrm{lp}} 
+ \boldsymbol{P}_{\mathrm{lp}}\boldsymbol{A}_{\mathrm{lp}}(\boldsymbol{K})
&
\boldsymbol{P}_{\mathrm{lp}}\boldsymbol{B}_{\mathrm{lp}}
&
\boldsymbol{C}_{\mathrm{lp}}(\boldsymbol{K})^\top
\\[3pt]
\boldsymbol{B}_{\mathrm{lp}}^\top \boldsymbol{P}_{\mathrm{lp}}
&
-\gamma_\infty^{2}\,\boldsymbol{I}
&
\boldsymbol{0}
\\[3pt]
\boldsymbol{C}_{\mathrm{lp}}(\boldsymbol{K})
&
\boldsymbol{0}
&
-\,\boldsymbol{I}_{n^{HP}}
\end{bmatrix}
\prec \boldsymbol{0}.
\end{equation}}

The BRL-based low-pass performance BMI~\eqref{eq:BRL-lp}  is equivalent to 
\small{\begin{equation}\label{eq:LMI_lp_direct}
\boxed{\;
\begin{bmatrix}
\boldsymbol{\Psi}_{11} & \boldsymbol{B}_{\rm lp} & \boldsymbol{\Psi}_{13}\\
\boldsymbol{B}_{\rm lp}^\top & -\gamma_\infty^{2}\,\boldsymbol{I} & \boldsymbol{0}\\
\boldsymbol{\Psi}_{31} & \boldsymbol{0} & -\boldsymbol{I}_{n^{HP}}
\end{bmatrix}\prec 0,\quad
\boldsymbol{X}\succ0,\ q_\omega>0,\;}
\end{equation}}

\noindent where $\boldsymbol{\Psi}_{11}$,
$\boldsymbol{\Psi}_{13}$,
and $\boldsymbol{\Psi}_{31}:=\boldsymbol{\Psi}_{13}^\top$ are defined in Appendix.

Combining the KYP-based passivity LMI~\eqref{eq:LMI_pass_direct}
and the BRL-based  LMI~\eqref{eq:LMI_lp_direct},
the controller design can be formulated as:
\begin{equation}\label{eq:joint-SDP}
\boxed{
\begin{aligned}
\min_{\boldsymbol X,\boldsymbol Y,q_\omega,\rho,\gamma_\infty}& \gamma_\infty^2\\[3pt]
\text{s.t. } & ~\eqref{eq:LMI_pass_direct}, ~\eqref{eq:LMI_lp_direct}.
\end{aligned}}
\end{equation}
Upon feasibility the single semidefinite program (SDP), the controller is recovered as
$\boldsymbol K=\boldsymbol Y\boldsymbol X^{-1}$.

\section{Numerical Experiments}
To validate the proposed electro–thermal control strategy, we conduct simulations on the modified Barry Island test system\cite{yi2023energy}. The testbed consists of a reduced 33-bus EPS and a 33-node DHS, with three HPs installed at buses/nodes \{1,\,32,\,33\}.

\subsection{DHS properties of the proposed temperature regulator}
\subsubsection{Disturbance-independent DHS Regulator}
We implement the augmented DHS~\eqref{augmentdhs} with the integral state $\boldsymbol{\xi}$ enforcing the optimality condition, so no disturbance forecasting is required. The closed-loop matrix $\boldsymbol{A_{\mathrm{cl}}}$ is Hurwitz, ensuring a unique equilibrium and $\boldsymbol{e}(t)\!\to\!0$ for any bounded disturbance.  
The DHS is tested with $V=50$\,L\footnote{A uniform water volume of $50$\,L per node sets the DHS thermal time scale, yielding closed-loop temperature constants of a few hundred seconds (Fig.~\ref{lf}).} under three heat-load profiles: a constant 0.1\,MW, a 0.2\,MW spike at $t=10$\,s, and decaying white noise. In all cases, temperatures and tracking errors converge, and the optimal equilibrium $(\boldsymbol{T^\ast},\boldsymbol{h^{G\ast}})$ is recovered (Fig.~\ref{lf}\footnote{Temperatures are shown as deviations from nominal; thermal inertia keeps them in the millikelvin range.}).  
These results confirm that the disturbance-independent regulator achieves real-time optimal temperature control, validating \textbf{Lemma~3} and \textbf{Lemma~4}.

\vspace{-1.2cm}
  \begin{figure}[htbp]
	\vspace{0.8cm}
	\hspace{-0.1cm}
	\includegraphics[width=3.2in]{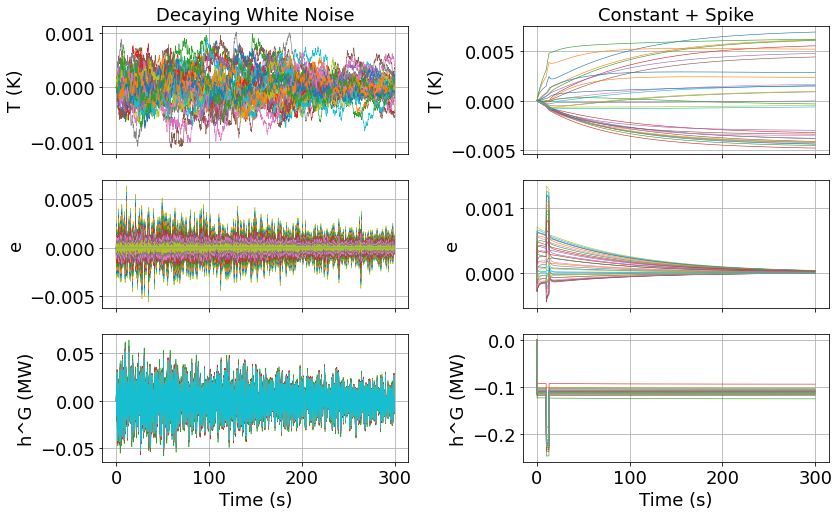}
	\vspace{-0.35cm}
	\caption{DHS with differenent types of disturbances.}
	\label{lf}
\end{figure}
\vspace{-0.3cm}

\subsubsection{MIMO Frequency Response of the DHS Port} 
Fig.~\ref{fr} shows the frequency response of the map $\boldsymbol{\omega^{HP}\!\to p^H}$ for lumped volumes $V\!=\!\{30,50,70\}L$. The 
eigenloci (bottom) confirm that both regulators remain positive real across all frequencies. For the passivity-only design (blue dashed), smaller water volumes accelerate DHS dynamics and create a distinct mid-frequency peak in $\sigma_{\max}(\boldsymbol{G(j\omega)})$, reflecting phase lag as disturbances enter the DHS time scale; this peak weakens and eventually disappears as $V$ increases. In contrast, the joint-OPT controller (orange) suppresses this effect for all volumes, producing a flatter gain profile and an almost vertical 
eigenlocus—evidence of consistently well-damped behavior.
\vspace{-1.0cm}
  \begin{figure}[htbp]
	\vspace{0.4cm}
	\hspace{0.3cm}
	\includegraphics[width=3.2in]{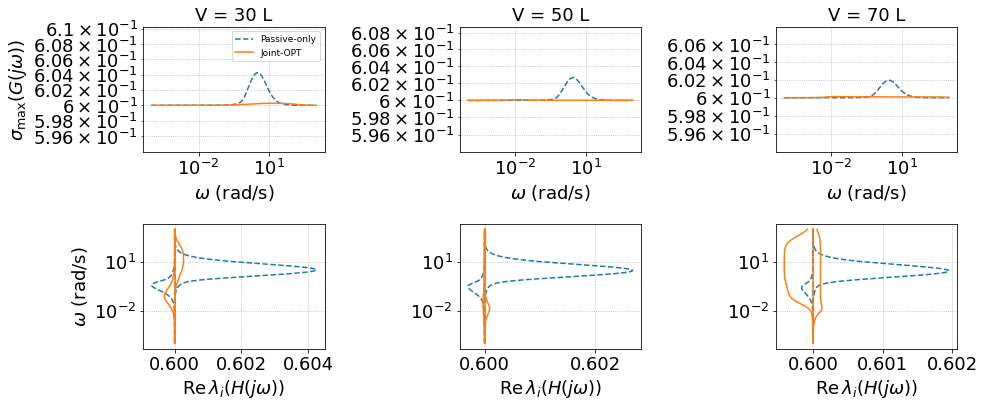}
	\vspace{-0.25cm}
	\caption{Frequency response $\boldsymbol{\omega^{HP}\!\to\!p^H}$.}
	\label{fr}
\end{figure}
\vspace{-0.6cm}

\subsection{Closed-loop CHP behavior}
\subsubsection{Stability and optimality of the CHP system}We test the robustness of the proposed joint  
$\mathcal{H}_{\infty}$--passivity controller under four coupling settings  
$(\gamma^E,\gamma^H)\!\in\!\{(0,0),(0,0.6),(0.6,0),(0.6,0.6)\}$ (Fig.~\ref{comparechp}), separating the roles of electrical support ($\gamma^E$) and thermal responsiveness ($\gamma^H$). Each case uses the same synthesis procedure. All simulations converge to the desired equilibrium,  
$\boldsymbol{w\!\to\!0,\; e\!\to\!0,\; h^G\!\to h^{G\ast},\; T\!\to T^\ast}$,  
confirming that the controller meets the steady-state objectives of both systems.

\subsubsection{The functional roles of the coupling gains} Program:chpcompare

Fig.~\ref{comparechp} highlights the roles of the coupling gains.  
A nonzero $\gamma^H$ sharply reduces DHS tracking error by allowing the HP–DHS subsystem to counteract temperature deviations rather than passively absorbing EPS disturbances (see (a)–(b) and (c)–(d)).  
Similarly, $\gamma^E>0$ improves EPS transients by enabling HPs to provide direct frequency support, reducing overshoot and speeding recovery; when $\gamma^E=0$, the EPS must respond alone, yielding larger deviations.

\subsubsection{The influence of loop shaping}
Table~\ref{tab:freq_ratio_full} reports the joint/ns ratios at three representative frequencies.  
For small $\omega$, both controllers exhibit nearly identical virtual damping.  
Near the DHS thermal eigenfrequency ($0.016$\,Hz), the DHS ceases to provide damping and becomes a disturbance-transmission path. The joint controller suppresses $\sigma_{\max}(\boldsymbol{G(j\omega)})$ via $\mathcal{H}_\infty$ low-pass loop shaping, reducing cross-domain coupling and improving EPS performance.  
At large $\omega$, DHS dynamics roll off naturally, and the joint design steepens this roll-off, further limiting high-frequency transmission.

\vspace{-1.2cm}
  \begin{figure}[htbp]
	\vspace{0.8cm}
	\hspace{-0.3cm}
	\includegraphics[width=3.5in]{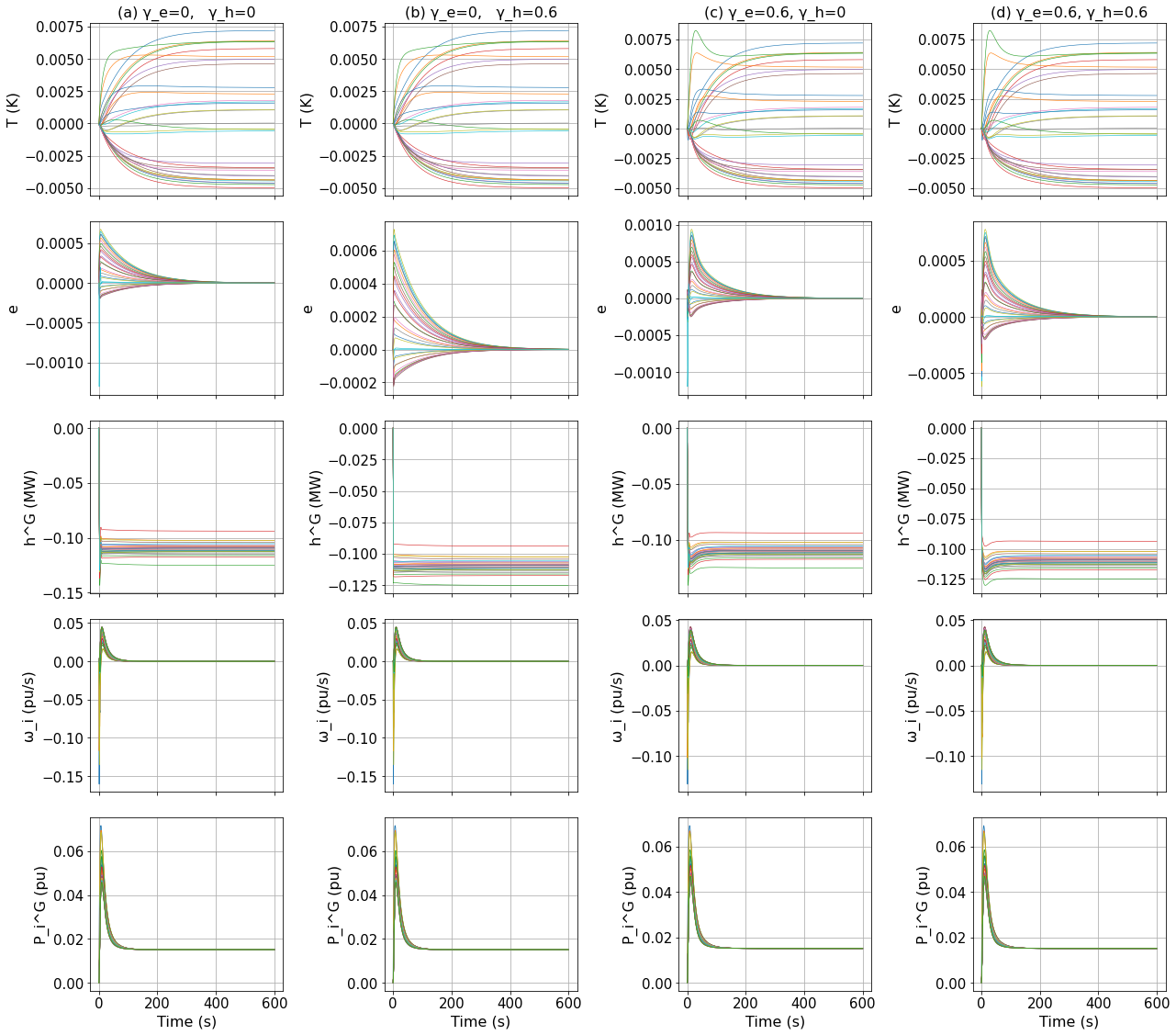}
	\vspace{-0.25cm}
	\caption{CHP system operation with different $\gamma^E$ and $\gamma^H$.}
	\label{comparechp}
\end{figure}
\vspace{-0.3cm}

\vspace{-0.3cm}
\begin{table}[htbp]
\centering
\vspace{-0.2cm}
\caption{Joint/passivity-only performance ratio at representative disturbance frequencies  
(values $<1$ indicate improvement over passive-only).}
\label{tab:freq_ratio_full}
\begin{tabular}{c|cccc}
\hline
Frequency 
& $J_{\mathrm{freq},L1}$ 
& $J_{u,L2}$ 
& peak $|h^G_k|$ 
& $J_{h^G,L2}$ \\
\hline
$0.0016$ Hz 
& $0.24$ 
& $0.11$ 
& $0.96$ 
& $1.00$ \\

$0.016$ Hz 
& $0.006$ 
& $0.021$ 
& $0.93$ 
& $0.99$ \\

$50$ Hz 
& $0.008$ 
& $0.019$ 
& $0.96$ 
& $1.00$ \\
\hline
\end{tabular}
\end{table}
\vspace{-0.5cm}

\section{conclusion}
This work presents a mixed $\mathcal{H}_\infty$--passivity framework that coordinates DHS temperature regulation with EPS secondary frequency control.
A forecast-free DHS regulator achieves the optimal energy-sharing equilibrium without heat-load prediction and provides tunable EPS–DHS performance trade-offs via LMI conditions.
The resulting controller guarantees closed-loop stability and convergence to the optimal CHP operating point, and simulations across low-, mid-, and high-frequency disturbances confirm the expected performance gains.

\bibliographystyle{IEEEtran}
\bibliography{CDC-port}
\appendix
\textbf{Proof of Lemma 1:} At equilibrium, $\dot g = 0$ in~\eqref{eq:agc} implies $\omega_r = 0$, and therefore
$\omega_i^{*} = 0$ for all $i$ under \textbf{Assumption~1}.

\textbf{Proof of Lemma 2:} Under \textbf{Assumption 2},   
classical passivity results for swing equations with strictly positive real generation dynamics imply that the overall EPS is 
input strictly passive from $-p_i^H$ to $\omega_i$. Restricting the input to the HP bus
frequencies as $\boldsymbol{y_1=\omega^{HP}}$ yields the dissipation inequality
\eqref{eq:eps_passive_ineq} for some storage function $V_{\mathrm{e}}$
and some $\rho_e>0$ depending on $D_{\min}$ and the margins of $G_i(s)$, proving strict passivity from $\boldsymbol{u_1=-p^H}$ to 
$\boldsymbol{y_1=\omega^{HP}}$.

\textbf{Proof of Lemma 4:}
Since $\boldsymbol{A}_{\rm cl}$ is Hurwitz, the closed-loop system has a unique equilibrium $\boldsymbol{x}^\ast$, and the state converges as $\boldsymbol{x}(t)\to \boldsymbol{x}^\ast$ for $t\to\infty$.  
At equilibrium, $\dot{\boldsymbol{\xi}}=\boldsymbol{e}= \boldsymbol{0}$, which implies $\lim_{t\to\infty} \boldsymbol{e}(t) = \boldsymbol{0}$.

\textbf{Proof of Lemma 5:}
\textbf{Lemma~4} guarantees $\lim_{t\to\infty}\boldsymbol e(t)=\boldsymbol 0$,  
which enforces the DHS optimality condition (\textbf{Lemma~3}),  
while \textbf{Lemma~1} ensures $\lim_{t\to\infty}\boldsymbol \omega(t)=\boldsymbol 0$ in the EPS.

\textbf{Proof of Theorem 1:}
Since the augmented DHS with temperature regulator~\eqref{errorde2}--\eqref{closeloopdhs}  
is strictly passive with respect to  
$\boldsymbol{u_2 = \omega^{HP}}$ and $\boldsymbol{y_2 = p^H}$,  
there exist a storage function $V_h$  
and a constant $\rho_h>0$ such that  
$\dot{V}_h 
\;\le\;
\boldsymbol{(p^H)^\top \omega^{HP}}
\;-\;
\rho_h\, \|\boldsymbol{p^H}\|^2$.
Summing the two storage inequalities yields  
$\dot{V} 
= \dot{V}_e + \dot{V}_h
\;\le\;
-\,\rho_h \|\boldsymbol{p^H}\|^2 -\,\rho_e \|\boldsymbol{\omega^{HP}}\|^2
\;\le\; 0$,
so $V := V_e + V_h$ is nonincreasing (with $V_e$ defined in (\ref{eq:eps_passive_ineq})).  
The largest invariant set where $\dot{V}=0$ requires $\boldsymbol{p^H=0}$ and $\boldsymbol{\omega^{HP}=0}$. 
By LaSalle’s invariance principle, this equilibrium is asymptotically stable.

\textbf{Proof of Lemma 6:}
The original DHS~(\ref{originaldhs}) is dominated by thermal storage and diffusive heat transport, which attenuate high-frequency disturbances; thus, it inherently acts as a low-pass system. For the augmented DHS~\eqref{errorde2}--\eqref{closeloopdhs}, the integral state $\boldsymbol{\dot\xi=e}$ and feedback $\boldsymbol{h^G=-K_T T-K_I\xi}$ ensure steady-state $\boldsymbol{e^\ast=0}$ under bounded disturbances whenever $\boldsymbol{A_{\rm cl}}$ is Hurwitz (see \textbf{Lemma~4}). Therefore,    $\boldsymbol{p^{H\ast}=\gamma^E\omega^{HP\ast}}$, and $\lim_{\omega\to0}\boldsymbol{G(j\omega)}=\boldsymbol{\gamma^E}$.

\textbf{KYP BMI\eqref{eq:KYP-BMI-strict}:} Let $\boldsymbol{X}:=\boldsymbol{P^{-1}}\succ0$,
applying the congruence transformation with $\boldsymbol{T_R}=\mathrm{diag}(\boldsymbol X,\boldsymbol I)$ to
\eqref{eq:KYP-BMI-strict} yields
$\boldsymbol{T_R}^\top
\begin{bmatrix}
\boldsymbol A_{\rm cl}^\top \boldsymbol P + \boldsymbol P\,\boldsymbol A_{\rm cl}
&
\boldsymbol P\,\boldsymbol B_{\rm cl}^{(w)} - \boldsymbol C_y(\boldsymbol K)^\top
\\[3pt]
*
&
-\bigl(\boldsymbol\gamma + \boldsymbol\gamma^\top\bigr) 
-\rho \boldsymbol I
\end{bmatrix}
\boldsymbol{T_R}
\ \prec\ \boldsymbol 0$. The left--upper block becomes $\boldsymbol X\big(\boldsymbol A_{\rm cl}^\top \boldsymbol P
      + \boldsymbol P\,\boldsymbol A_{\rm cl}\big)\boldsymbol X =\boldsymbol X \boldsymbol A_{\rm cl}^\top
      + \boldsymbol A_{\rm cl}\boldsymbol X
     =\operatorname{sym}(\boldsymbol A_{\rm cl}\boldsymbol X) \prec 0$.
With $\boldsymbol A_{\rm cl}=\boldsymbol A_{\rm aug}-\boldsymbol B_{\rm aug}\boldsymbol K$, $\operatorname{sym}\big((\boldsymbol A_{\rm aug}-\boldsymbol B_{\rm aug}\boldsymbol K)\boldsymbol X\big)
=
\operatorname{sym}\big(\boldsymbol A_{\rm aug}\boldsymbol X - \boldsymbol B_{\rm aug}(\boldsymbol K\boldsymbol X)\big)\prec0$.
Introduce the standard change of variables
$\boldsymbol Y := \boldsymbol K \boldsymbol X$
to obtain 
$\boldsymbol{LMI_1}=\operatorname{sym}\!\big(\boldsymbol A_{\rm aug}\boldsymbol X
- \boldsymbol B_{\rm aug}\boldsymbol Y\big) \prec 0,\boldsymbol X\succ0$. For the off--diagonal block, the congruence transformation yields
$\boldsymbol X(\boldsymbol P\,\boldsymbol B_{\rm cl}^{(w)} 
- \boldsymbol C_y(\boldsymbol K)^\top)
= \boldsymbol B_{\rm cl}^{(w)} 
- \boldsymbol X \boldsymbol S_C^\top
+ \boldsymbol X \boldsymbol K^\top \boldsymbol S_D^\top$.
With 
$\boldsymbol{LMI_2}
:= \boldsymbol B_{\rm cl}^{(w)} 
- \boldsymbol X \boldsymbol S_C^\top
+ \boldsymbol Y^\top \boldsymbol S_D^\top$,
 ~\eqref{eq:KYP-BMI-strict} is equivalently written as the LMI (\ref{eq:LMI_pass_direct}).

\textbf{BRL BMI\eqref{eq:BRL-lp}:} $\boldsymbol{P}_{\mathrm{lp}}=\mathrm{diag}(\boldsymbol{P},\,p_w\boldsymbol{I}_{n^{HP}})\succ\boldsymbol{0}$ and define
$\boldsymbol{X}:=\boldsymbol{P}^{-1}\succ\boldsymbol{0}$, 
$\;q_w:=p_w^{-1}>0$, and $\boldsymbol{X}_{\mathrm{lp}}:=\boldsymbol{P}_{\mathrm{lp}}^{-1}=\mathrm{diag}(\boldsymbol{X},\,q_w\boldsymbol{I}_{n^{HP}})$.
With $\boldsymbol{TR}_{\mathrm{lp}}=\mathrm{diag}(\boldsymbol{X}_{\mathrm{lp}},\,\boldsymbol{I},\,\boldsymbol{I})$,  \eqref{eq:BRL-lp} becomes
\small{\begin{equation}\notag
\begin{bmatrix}
\underbrace{\boldsymbol{A}_{\mathrm{lp}}(\boldsymbol{K})\,\boldsymbol{X}_{\mathrm{lp}}
+\boldsymbol{X}_{\mathrm{lp}}\boldsymbol{A}_{\mathrm{lp}}(\boldsymbol{K})^\top}_{\triangleq~\boldsymbol{\Psi}_{11}}
&
\boldsymbol{B}_{\mathrm{lp}}
&
\underbrace{\boldsymbol{X}_{\mathrm{lp}}\boldsymbol{C}_{\mathrm{lp}}(\boldsymbol{K})^\top}_{\triangleq~\boldsymbol{\Psi}_{13}}
\\
\boldsymbol{B}_{\mathrm{lp}}^\top
&
-\gamma_\infty^{2}\,\boldsymbol{I}
&
\boldsymbol{0}
\\
\underbrace{\boldsymbol{C}_{\mathrm{lp}}(\boldsymbol{K})\,\boldsymbol{X}_{\mathrm{lp}}}_{\triangleq~\boldsymbol{\Psi}_{31}}
&
\boldsymbol{0}
&
-\,\boldsymbol{I}_{n^{HP}}
\end{bmatrix}
\prec \boldsymbol{0}.
\end{equation}}
Define
$\boldsymbol{A}_{\sim}:=\boldsymbol{A}_{\mathrm{aug}}\boldsymbol{X}-\boldsymbol{B}_{\mathrm{aug}}\boldsymbol{Y},
\boldsymbol{C}_{\sim}:=\boldsymbol{S}_C\boldsymbol{X}-\boldsymbol{S}_D\boldsymbol{Y}$.
Then with
$\boldsymbol{\Psi}_{11}
\begin{bmatrix}
\boldsymbol{A}_{\sim}+\boldsymbol{A}_{\sim}^\top & \omega_h\,\boldsymbol{C}_{\sim}^\top\\
\omega_h\,\boldsymbol{C}_{\sim} & -2\,\omega_h\,q_w\,\boldsymbol{I}_{n^{HP}}
\end{bmatrix}$,
$\boldsymbol{\Psi}_{13}
=
\begin{bmatrix}
(1+\alpha)\big(\boldsymbol{X}\boldsymbol{S}_C^\top-\boldsymbol{Y}^\top\boldsymbol{S}_D^\top\big)\\
-\,q_w\,\boldsymbol{I}_{n^{HP}}
\end{bmatrix}$,
$\boldsymbol{\Psi}_{31}=\boldsymbol{\Psi}_{13}^\top$, the BRL BMI is equivalently written as the LMI (\ref{eq:LMI_lp_direct}).
Upon feasibility, recover the controller by
$\ \boldsymbol{K}=\boldsymbol{Y}\boldsymbol{X}^{-1}$.

\end{document}